\documentclass[final, reprint, notitlepage, narroweqnarray,
	inline,
	twoside,
	]{ieee}

        \usepackage{ieeefig}

\usepackage[T1]{fontenc}
\usepackage{inputenc}
\usepackage{url}

\usepackage{multicol}

\usepackage{amsmath}

\usepackage{amsfonts}
\usepackage{indentfirst}

\usepackage{graphicx}
\usepackage{xspace}
\usepackage{epsfig}
\usepackage{times}
\usepackage{color}
\usepackage{mathrsfs,multirow}
\usepackage[boxed]{algorithm2e}
\usepackage{ae}
\usepackage{amsmath}
\usepackage{amssymb}
\usepackage{theorem}

\makeatletter

\usepackage{color,dsfont}

\usepackage{latexsym, amssymb, amsmath, tabularx, epsfig, xspace}

\newcommand{\F}{\ensuremath{\mathbb{F}_2}}

\newcommand{\prob}[1]{\ensuremath{\mathsf{Pr}\left[#1\right]}}

\newtheorem{theorem}{Theorem}
\newtheorem{lemma}{Lemma}
\theoremstyle{plain}\newtheorem{definition}{Definition}[section]

\newcommand{\SIGscheme}{\schemefont{S}}

\newcommand{\Skg}{\algfont{KG}}
\newcommand{\Ssign}{\algfont{Sign}}

\newcommand{\Svf}{\algfont{Vfy}}

\newcommand{\IBSIGscheme}{\schemefont{IBS}}
\newcommand{\IBSkg}{\algfont{MKG}}

\newcommand{\IBSderive}{\algfont{UKg}}

\newcommand{\IBIscheme}{\schemefont{IBI}}
\newcommand{\IBIKg}{\algfont{MKg}}
\newcommand{\IBIP}{\algfont{\overline{P}}}
\newcommand{\IBIV}{\algfont{\overline{V}}}
\newcommand{\IBIderive}{\algfont{UKg}}
\def\CP{\overline{\mathsf{CP}}}
\def\CV{\overline{\mathsf{CV}}}

\newcommand{\heading}[1]{{\vspace{5pt}\noindent\sc{#1}}}

\DeclareMathAlphabet{\mathscript}{OT1}{pzc}{m}{it}
\newcommand{\schemefont}[1]{{\mathscript{#1}}}
\newcommand{\algfont}[1]{{\mathsf{#1}}}

\newcommand{\varfont}[1]{\mathit{#1}}

\def\sig{\varfont{sig}}

\def\PK{\varfont{PK}}
\def\PK{\varfont{PK}}
\def\SK{\varfont{SK}}
\def\msk{\varfont{msk}}
\def\mPK{\varfont{mpk}}
\def\sk{\varfont{SK}}

\def\UKlist{UK^{\hbox{list}}}
\def\hlist{h^{\hbox{list}}}
\def\Llist{\Lambda}

\newcommand{\ID}{{\cal ID}}
\newcommand{\id}{\mathit{id}}

\def\Binary{\mathsf{Binary}}
\def\Goppa{\mathsf{Goppa}}
\def\IDch{id^\star}

\def\N{\mathbb{N}}

\def\A{\mathcal{A}}
\def\C{\mathcal{C}}
\def\P{\mathsf{P}}
\def\F{\mathbb{F}}
\def\V{\mathsf{V}}
\def\acc{\mathsf{acc}}
\def\rej{\mathsf{rej}}
\def\Kg{\mathsf{Kg}}

\def\D{\mathcal{D}}
\def\O{\mathcal{O}}
\def\Run{\mathbf{Run}}
\def\getsr{\stackrel{R}{\leftarrow}}

\begin{document}
\title[Improved identity-based identification using correcting codes]{%
       Improved identity-based identification using correcting codes}

\author[CGGG]{
      Pierre-Louis Cayrel$^1$ 
      \authorinfo{
      1 - Université de Paris 8, LAGA,
      Département de Mathématiques,
      2, rue de la liberté,
      93526 Saint-Denis cedex 02, France,
      email: \mbox{cayrelpierrelouis@gmail.com}}
      \and
      Philippe Gaborit$^2$ 
      \authorinfo{
      2 - Université de Limoges, XLIM-DMI,
      123, Av. Albert Thomas 87060 Limoges Cedex France,
      email: \mbox{philippe.gaborit@xlim.fr}}
      \and
      David Galindo$^3$
      \authorinfo{
            3 - University of Luxembourg 6, rue Richard Coudenhove-Kalergi
L-1359 Luxembourg
      email: \mbox{david.galindo@uni.lu}}
      \and
      and Marc Girault$^4$ 
      \authorinfo{
      4 - Orange Labs
      42, rue des Coutures 14066 Caen France,
      email: \mbox{marc.girault@orange-ftgroup.com}}
  }
\journal{IEEE Trans.\ on Information\ Theory}
\firstpage{1}

\maketitle

\begin{abstract}
In this paper, a new identity-based identification scheme based on error-correcting codes is proposed.

Two well known code-based schemes are combined~: the signature scheme by Courtois, Finiasz and Sendrier and an identification scheme by Stern.

A proof of security for the scheme in the Random Oracle Model is given.

\end{abstract}


\section{Introduction}
\PARstart One of the most critical points of public key cryptography (PKC) is
that of the management of the authenticity of the public key. It is the very single point that anchors public key cryptography to the real world. If no such a mechanism is provided the consequences are fatal. In fact, if Alice is able to take Bob's identity by faking her own
public key as Bob's one, she would be able to decipher all messages
sent to Bob or to sign any message on behalf of Bob.

\indent In 1984, Shamir introduced the concept of Identity-based Public Key Cryptography ID-PKC
\cite{AS} in order to simplify the management and the identification of
the public key, which, time passing by, had become more and more complex.

\indent In ID-PKC the public key of an user is obtained from his identity $\id$ on the network. The identity $\id$ can be a
concatenation of any publicly known information that singles out the user~: a name, an e-mail,
or a phone number, to name a few.
Hence it is not longer necessary to verify a certificate for the public key
nor to access a public directory to obtain a certificate.
At first glance it seems simple but producing private keys becomes
more complex. In particular a user can not build his own
private key by himself anymore, and it is necessary to introduce a
trusted third party who constructs the private key from the user's identity and sends it to the user. This process has to be done at least once for each
user.

\indent Shamir \cite{AS}  calls this trusted third party the Key Generation
Center (KGC). The KGC is the owner of a system-wide secret, thus called the \emph{master key}. After
successfully verifying (by non-cryptographic means) the identity of the user, the KGC
computes the corresponding user private key from the master key, the user identity $\id$ and a trapdoor function.

\indent Identity-based systems resemble ordinary public-key systems, in the sense that both involve a private transformation (i.e. decrypting) as well as a
public transformation (i.e. encrypting).  However, in identity-based systems users do not have explicit public keys. Instead,
the public key is effectively replaced by (or constructed from) a user's publicly available
identity information.

The motivation behind identity-based systems is to create a cryptographic system resembling
an ideal mail system.  In this ideal system, knowledge of a person's name alone suffices for confidential mailing to that person, and for signature verification that only that person
could have produced. In such an ideal cryptographic system~:
\begin{enumerate}
\item  users need not exchange neither symmetric keys nor public keys;
\item  public directories (databases containing public keys or certificates) need not be kept;
\item  the services of a trusted authority are needed solely during a set-up phase (during
which users acquire authentic public system parameters).
\end{enumerate}

\medskip

A drawback in many concrete proposals of identity-based systems is that the required user-specific identity data includes additional data, taking the form of an integer or public data value for instance, denoted DA, beyond an a priori identity ID. Ideally, DA is not required, as a primary motivation for identity-based schemes is to eliminate the need to transmit public keys, to allow truly non-interactive protocols with identity information itself sufficing as an authentic public key. We will refer to the latter systems as \emph{pure identity-based} systems. The issue is less significant in signature and identification
schemes where the public key of a claimant is not required until receiving a message
from that claimant (in this case DA is easily provided); but in this case, the advantage
of identity-based schemes diminishes. It is more critical in key agreement and public-key
encryption applications where another party's public key is needed at the outset.

\indent In his paper Shamir proposed identity-based signature and identification systems based on the RSA or Discrete Logarithm problems.
The first efficient provably secure identity-based encryption cryptosystem featuring the above mentioned non-interactive property  was proposed in 2001
by Boneh and Franklin \cite{MFDB}. This system is based on the Weil pairing
over certain families of elliptic curves. The same year, Cocks \cite{CC}
published a system based on quadratic residuosity but a rather large message expansion makes it somewhat inefficient in practice.

\medskip

\indent Following the paper by Boneh and Franklin, research on ID-PKC has made great advances
and lots of schemes have been published, most of them based on
elliptic curves and bilinear pairings, such as identity-based encryption
(IBE) schemes \cite{5}, identity-based key agreement
schemes \cite{6}, identity-based identification (IBI) or identity-based signature (IBS) schemes \cite{11,13,14}. 
In 2004  Bellare, Neven and Namprempre
proposed in \cite{BNN} a general framework deriving
IBI or IBS from traditional public key-based signature and identification schemes and
they applied it to concrete known schemes. The resulting systems are not pure identity-based and only schemes based on number theoretic problems were considered.

In this paper, we propose and formally study a new IBI scheme built from error-correcting codes.

Code-based cryptography was introduced by McEliece \cite{McE78}, a variation of which was later proposed by Niederreiter \cite{Nied}.
The idea of using error-correcting codes for identification purposes
is due to Harari \cite{har88}, followed by Stern (first protocol) and Girault \cite{gir90}.
But Harari and Girault protocols were subsequently broken, while Stern's
one was five-pass and unpractical.
At Crypto'93, Stern proposed a new scheme \cite{Ste93}, which is still today the
reference in this area.

For a long time no code-based signature scheme was known, eventually the first (not yet cryptanalyzed) one was proposed by Courtois, Finiasz and Sendrier \cite{CFS} (CFS) in 2001.
The basic idea of the CFS signature scheme is to choose parameters
such that an inversion of the otherwise non-invertible Niederreiter scheme is feasible. This is done at the cost of a rather large public key when comparing to
other signature schemes. Still signature length is short.

We obtain our new IBI scheme by combining the CFS signature scheme and the identification
scheme by Stern. The basic idea of our scheme is to start from a Niederreiter-like problem which can be inverted like in the CFS scheme. This permits to associate
a secret to a random (public) value obtained from the identity of the user.
The secret and public values are then used for the Stern
zero-knowledge identification scheme.

\indent The paper is organized as follows.
In Section \ref{sec:defs} we introduce notation and definitions,
while in Section \ref{Pb} we recall basic facts on code-based cryptography. Section \ref{CFS} is devoted to describe the public key encryption scheme of Niederreiter and the signature scheme of Courtois, Finiasz and Sendrier. The identification protocol of Stern is presented in Section \ref{Stern}, and next   our new protocol is described in Section \ref{scheme}.
In Section \ref{sec:proof} we give a proof of security for our scheme in the Random Oracle Model \cite{BR}.

Finally in Section \ref{sec:efficiency}
we give concrete parameters and conclude in Section \ref{sec:conclusion}.
\medskip

\noindent\textbf{Publication info.} This is the full version of a previously publish conference extended abstract \cite{CGG07}.

\section{Notation and definitions} \label{sec:defs}
\PARstart We first introduce some notation. If $x$ is a string, then $|x|$ denotes its length,
while if $S$ is a set then $|S|$ denotes its cardinality. If $\kappa \in \N$ then $1^\kappa$ denotes the
string of $\kappa$ ones.

\medskip

If $S$ is a set then $s \getsr S$ denotes the operation of picking an element $s$ in $S$
uniformly at random. Unless otherwise indicated, algorithms are modelled as Probabilistic
Polynomial Time (PPT) algorithms. We write $\A(x,y, \ldots)$ to indicate that $\A$ is an
algorithm with inputs $x,y, \ldots $ and by $z \gets \A(x,y, \ldots)$ we denote the
operation of running $\A$ with inputs $(x,y, \ldots )$ and letting $z$ be the output. We
write $\A^{\O_1, \O_2, \ldots}(x,y, \ldots)$ to indicate that $\A$ is an algorithm with
inputs $x,y, \ldots $ and access to oracles $\O_1, \O_2, \ldots$ and by $z \gets
\A^{\O_1, \O_2, \ldots}(x,y, \ldots)$ we denote the operation of running $\A$ with inputs
$(x,y, \ldots )$ and access to oracles $\O_1, \O_2, \ldots$ and letting $z$ be the
output.

\medskip

\heading{Provers and verifiers.} An interactive algorithm is a stateful PPT algorithm
that on input an incoming message $M_{\hbox{in}}$ (this is $\varepsilon$ if the party is
initiating the protocol) and state information $St$ outputs an outgoing message
$M_{\hbox{out}}$ and updated state $St$. The initial state contains the initial inputs
of the algorithm. We say that $\A$ accepts if $M_{\hbox{out}}=\acc$ and rejects if $M_{\hbox{out}}=\rej$. An
interaction between a prover $\P$ and a verifier $\V$, both modelled as interactive
algorithms, ends when $\V$ either accepts or rejects. The expression~:

$$(C,d) \leftarrow \Run [\P(p_1,\ldots) \leftrightarrow\V(v_1,\ldots)]$$

denotes that $\P$ and $\V$ have initiated in an interaction with inputs
$p_1,\ldots$ and $v_1,\ldots$ respectively, getting a conversation transcript $C$ and a
boolean decision $d$, with 1 meaning that $\V$ accepted, and 0 meaning it rejected.

\medskip

\heading{Standard identification schemes.} A standard identification scheme $\SIGscheme=
(\Kg, \P, \V)$ consists of three PPT algorithms~:

\begin{itemize}
\item[{\bf Key generation}] algorithm $\Kg$ takes as input a security parameter $\kappa$ and returns
a secret key $\SK$ and a matching public key $\PK$. We use the notation $(\SK, \PK)
\leftarrow \Kg(1^\kappa)$.
\item[{\bf Interactive identification}] protocol, where the prover runs $\P$ with initial
state $\SK$, while the verifier has initial state $\PK$. It is required that for all $\kappa
\in\N$ and valid key pairs $(\PK,\SK)$, the output by $\V$ in any interaction between
$\V$ (with input $\PK)$ and $\P$ (with input $\SK)$ is $\acc$ with probability one.
\end{itemize}

\medskip

\medskip

\heading{Standard Signatures.} A standard signature scheme $\SIGscheme= (\Skg, \Ssign,
\Svf)$ consists of three PPT algorithms~:

\begin{itemize}
\item[ {\bf Key generation}] algorithm $\Skg$ takes as input a security parameter $\kappa$ and returns
a secret key $\SK$ and a matching public key $\PK$. We use the notation $(\SK, \PK)
\leftarrow \Skg(1^\kappa)$.
\item[{\bf Signing}] algorithm $\Ssign$ takes as input a secret key $\SK$ and a message $m$.
The output is a signature $\sig_\SK(m)$. This is denoted as $\sig_\SK(m) \leftarrow
\Ssign(\SK,m)$.
\item[{\bf Verification}] algorithm $\Svf$ takes as input a public key $\PK$, a message $m$,
and a signature $\sig=\sig_\SK(m)$.  The output is 1 if the signature is valid, or 0
otherwise.
We use the notation $\{0,1\} \leftarrow \Svf(\PK,m,\sig)$ to refer to one
execution of this algorithm.
\end{itemize}

\medskip

The standard security notion for signature schemes is unforgeability against
adaptively-chosen message attacks, which can be found in~\cite{GMR88}.

\medskip

\heading{Identity-Based identification.} \label{IBS} An identity-based identification scheme
$\IBSIGscheme= (\IBIKg, \IBIderive, \IBIP, \IBIV)$ consists of four PPT algorithms, as follows~:

\begin{itemize}
\item[{\bf Master-key generation}] algorithm $\IBIKg$ takes as input a security parameter $\kappa$ and returns, on one hand, the system public parameters $\mPK$ and, on the other hand, the matching
master secret key $\msk$, which is known only to a master entity. It is denoted as
$(\mPK, \msk) \leftarrow \IBSkg(1^\kappa)$.

\item[{\bf Key extraction}] algorithm $\IBSderive$ takes as inputs the master secret
key $msk$ and an identity $\id \in \{0,1\}^*$, and returns a secret key $\sk[\id]$. We
use the notation $\sk[\id] \leftarrow \IBSderive(\msk,\id)$.

\item[{\bf Interactive identification}] protocol, where the prover with identity $id$ runs the interactive algorithm $\IBIP$ with initial state $\sk[id]$, and the verifier runs $\IBIV$ with initial state $\mPK,id$.
\end{itemize}

\medskip

\heading{Security of IBI schemes.} An IBI scheme is said to be secure against impersonation under passive attacks ({\sf imp-pa}) if any adversary $\mathcal{A}=(\CP,\CV)$, consisting of a cheating prover $\CP$ and a cheating verifier $\CV$, has a negligible advantage in the following game~:

    \textit{Setup} The challenger takes a security parameter $\kappa$ and runs the master key generation algorithm $\IBIKg$. It gives $\mPK$ to the adversary and keeps the master secret key $\msk$ to itself. It initializes an empty list $\UKlist$.

    \textit{Phase 1} The adversary issues queries of the form
    \begin{itemize}
        \item[--] User key query $\langle \ID_i \rangle.$ The challenger checks whether         there exists an entry $(id_i,\sk[id_i])$ in the list  $\UKlist$. If this is the case, it retrieves the user secret key $\sk[id_i]$. Otherwise, it runs
        algorithm $\IBIderive$ to generate the private key $\sk[id_i]$ corresponding to $id_i$.
        It sends $\sk[id_i]$ to the adversary.         It includes the entry $(id_i,\sk[id_i])$ in the list  $\UKlist$.
        \item[--] Conversation query $\langle \ID_i \rangle.$ The challenger checks whether         there exists an entry $(id_i,\sk[id_i])$ in the list  $\UKlist$. If this is the case, it retrieves the user secret key $\sk[id_i]$. Otherwise, it runs
        algorithm $\IBIderive$ to generate the private key $\sk[id_i]$ corresponding to $\ID_i$. The challenger returns $(C,d)$ where $(C,d) \leftarrow \Run [\CP(\sk[id_i]) \leftrightarrow\V(\mPK,id_i)]$.
    \end{itemize}

    \medskip

    These queries may be asked adaptively, that is, each query may depend on the answers obtained to the
    previous queries.

    \medskip

    \textit{Challenge} The cheating verifier $\CV$ outputs a target
    identity $\IDch$ and its state $St_{\CV}$, such that the private key for $\IDch$ was not requested in Phase 1.

    \medskip

    \textit{Phase 2}  The cheating prover $\CP$, with input $St_{\CV}$,  interacts with a honest verifier with input $\mPK,\IDch$. The cheating prover is allowed to query the same oracles as in Phase 1, except that the query $\IDch$ is not allowed. Finally, $\A$ wins if the output of $\V$ is accept, i.e. $d=1$ in $(C,d) \leftarrow \Run [\CP(\sk[id_i]) \leftrightarrow\V(\mPK,id_i)]$.

    \medskip

    Such an adversary is called an {\sf imp-pa} adversary $\A$, and its advantage is defined as
    $\mathsf{Adv}^{\mathsf{imp-pa}}_{\IBIscheme,\mathcal{A}} (1^\ell)=\Pr[d=1].$

\bigskip

\section{Code-based cryptography}\label{Pb}
\PARstart In this section we recall basic facts about code-based cryptography.
We refer to the work of Sendrier \cite{S02b} for
a general introduction to these problems. 

\subsection{Hard problems}
Every public key cryptosystem relies on

a hard problem.
In the case of coding theory, the main hard problems used are the Bounded Decoding (BD) and Code Distinguishing (CD) problems.

\begin{definition}[Bounded Decoding  Problem]

  \label{def:BD}
Let $n$ and $k$ be two integers such that $n\geq k$ and $H$ a parity check matrix.
$\Binary(n,k)$ represents a random binary matrix of $n$ columns, $k$ rows and of rank $k.$
     \\
  \noindent\textsf{Input~:} $H \getsr \Binary(n,k)$ and
 $s\getsr\mathbb{F}_2^{n-k}$ \par\noindent\textsf{Ouput~: } A word
  $e\in\mathbb{F}_2^n$ such that $\textsf{wt}(e)\le \frac{n-k}{\log_2 n}$ and \\
  $He^T = s$
\end{definition}
  Let us denote by   $\mathsf{Adv}^{\mathsf{BD}}_{\mathcal{C}} (n,k)$ the probability that an algorithm $\C$ has in solving the above problem.

\noindent This problem was proven to be NP-complete in \cite{BMT78}.

\medskip

\begin{definition}[Code Distinguishing Problem]
 
  \label{def:GPDD}
  Let $n$ and $k$ be two integers such that $n\geq k$ and $H$ a parity check matrix. \\
  \noindent\textsf{Input~:} $H \getsr \Goppa(n,k)$ or $H \getsr \Binary(n,k)$.

  \par\noindent\textsf{Ouput~: } $b=1$ if $H\in \Goppa(n,k)$, $b=0$ otherwise.

\begin{eqnarray*}
\mathsf{Adv}^{\mathsf{CD}}_{\mathcal{D}} (n,k) &=& \big\lvert\Pr[\D(H)=1 \;|\; H \getsr \Goppa(n,k)]\\
& &-\Pr[\D(H)=1 \;|\; H \getsr \Binary(n,k)]\big\rvert.
\end{eqnarray*}
\end{definition}

The description of a Goppa code $\Goppa(n,k)$ of length $n$ and dimension $k$ is to be found in \cite{MS77}.

\subsection{McEliece scheme}
\medskip

\noindent [{\bf Key Generation}]
 Let $\mathcal{C}$ be a $q$-ary linear code $t$-correcting of length $n$
and of dimension $k.$ We denote $\mathcal{C}(n,k,d)$ a such code. Let $G$ a generator matrix of $\mathcal{C}.$ We will use an $G'$ matrix such that ~:
\begin{displaymath}
  G'=SHP\left\{
  \begin{array}{l}
    S \textrm{ is invertible}\\
    P \textrm{ is a permutation matrix}
  \end{array}\right.
\end{displaymath}
 $G'$ is public and its decomposition and a syndrome decoding algorithm for $\mathcal{C}$ are secret. To be clearer, we recall the various sizes of the matrices~:

 $S$ is $n-k\times n-k,$ $H$ is $n\times n-k,$ $P$ is $n\times n.$
\\

\noindent [{\bf Encryption}] Let $E_{q,n,t}$ bet the space of $\mathbb{F}_{q}^{n}$ words with Hamming weight $t$. For a chosen cleartext $x \in E_{q,n,t}$, $y$ is the cryptogram corresponding to $x$ if and only if $y=xG'+e.$\\

\noindent [{\bf Decryption}] For $y=xG'+e,$  the knowledge of the secret key allows ~:
\begin{enumerate}
\item to compute $u=yP^{-1},;$
\item to find $u'$ from $u$ thanks to a syndrome decoding algorithm;
\item to find $x=u'S^{-1}.$
\end{enumerate}

The syndrome decoding algorithm can be, for instance, in the case of Goppa's
codes, Patterson's algorithm (see part \ref{complexity}).

 \subsection{Cryptanalytic Attacks}
 The security of code-based cryptosystems depends on the difficulty of
the following two attacks~:
\begin{itemize}
\item[(i)] {\em{Structural Attack}} ~: Recover the secret transformation and the description
of the secret code(s) from the public matrix.
\item[(ii)] {\em{Ciphertext-Only Attack}}~: Recover the original message from the
ciphertext and the public key.
\end{itemize}

\subsubsection{Structural Attack}

\indent While no efficient algorithm for decomposing $G'$ into $(S, G, P)$ has been discovered yet
\cite{50}, a structural attack has been discovered in \cite{45}. This attack reveals part of the structure
of a so-called weak $G'$ where 'weak' means that $G'$ has been generated from a binary Goppa polynomial in a special manner. However, this attack can
be avoided simply by not using such weak public keys.

\medskip

Structural attacks
aim at recovering the structure of the
permuted code, i.e. recovering the permutation
from the code and its permuted version. The underlying problem
is the equivalence of codes. This problem was considered
by Sendrier for which he gave a nice solution~: the
Support Splitting Algorithm \cite{S02b}.

    \medskip

The complexity
of this algorithm is in $\mathcal{O}(2^{\text{dimension}(\mathcal{C} \cap \mathcal{C}^{\perp})})$ where $\mathcal{C}^{\perp}$ is the dual of the code $\mathcal{C}.$
This means that in order to resist the attack one gets two options~: either
starting from a large family of codes with arbitrary small hulls (the intersection
of $\mathcal{C}$ and $\mathcal{C}^{\perp}$) or starting from a small family of codes but with a large hull.

    \medskip

For instance the choice of Goppa codes corresponds to the first possibility.

\subsubsection{Ciphertext-Only Attack}
 A first analysis using the Information-Set-Decoding was done
by McEliece, then by Lee and Brickell, Stern  and Leon
and lastly by Canteaut and Chabaud (see \cite{CC2} for all references).

\medskip

The Information-Set-Decoding Attack is one of the known general attacks (i.e., not
restricted to specific codes) and seems to have the lowest complexity.

    \medskip

One tries to recover the
$k$ information symbols as follows~: the first step is to pick $k$ of the $n$ coordinates randomly in
the hope that none of the $k$ are in error. We then try to recover the message by solving the
$k \times k$ linear system (binary or over $\mathbb{F}_q$).
Let $G'_k, c_k$ and $z_k$ denote the $k$ columns picked from $G',c$ and $z,$ respectively. They have
the following relationship
$$c_k = mG'_k + z_k.$$
If $z_k = 0$ and $G'_k$ is non-singular, $m$ can be recovered by
$$m = c_k{G'}_{k}^{-1}.$$

\medskip

The computation cost of this version 
is $T(k) \times P_{n,k,t},$ where
$$P_{n,k,t} =\Pi_{i=0}^{k-1}(1-\frac{t}{n-i}).$$

The quantity $T(k)$ in the average work factor is the number of operations required to solve a
$k \times k$ linear system over $\mathbb{F}_q$. As mentioned in \cite{McE78}, solving a $k \times k$ binary system takes about
$k^3$ operations. Over $\mathbb{F}_q$, it would require at least $(k
\times \log_2 q)^3$ operations.

\medskip

All the papers which improve the complexity only impact the cost of the Gaussian elimination.
In the best improvement by Canteaut and Chabaud \cite{CC2} a good approximation
of the cost besides the probability factor can be taken roughly in $(k\times \log_2 q)^2$.

\indent Apart from these general attacks there are some attacks targeting McEliece cryptosystem
using specific codes(see \cite{SS92,45,5,Survey} for exemple).

\section{Signature scheme of Courtois, Finiasz and Sendrier (or CFS scheme)}
\label{CFS}

\PARstart Before describing the CFS scheme we first recall the Niederreiter public key cryptosystem.

\subsection{Niederreiter encryption scheme}

\noindent [{\bf Key Generation}] Let $\mathcal{C}$ be a binary linear code $t$-correcting of length $n$
and of dimension $k.$
Let $H$ a parity check matrix of $\mathcal{C}.$
We will use an $\widetilde{H}$ matrix such that ~:
\begin{displaymath}
  \widetilde{H}=QHP\left\{
  \begin{array}{l}
    Q \textrm{ is invertible}\\
    P \textrm{ is a permutation matrix}
  \end{array}\right.
\end{displaymath}
 $\widetilde{H}$ is public and its decomposition and a syndrome decoding algorithm for $\mathcal{C}$ are secret.

 To be clearer, we recall the various sizes of the matrices~:

 $Q$ is $n-k\times n-k,$ $H$ is $n\times n-k,$ $P$ is $n\times n.$

Let $E_{q,n,t}$ bet the space of $\mathbb{F}_{q}^{n}$ words with Hamming weight $t$.\\

\noindent [{\bf Encryption}] For a chosen cleartext $x$ in $E_{q,n,t}$,
$y$ is the cryptogram corresponding to $x$ if and only if $y=\widetilde{H}x^{T}.$\\

\noindent [{\bf Decryption}] For $y=\widetilde{H}x^{T},$ the knowledge of the secret key
allows ~:
\begin{enumerate}
\item to compute $Q^{-1}y\ (=HPx^{T});$
\item to find $Px^{T}$ from $Q^{-1}y$ thanks to a syndrome decoding algorithm;
\item to find $x$ applying $P^{-1}$ to $Px^{T}.$
\end{enumerate}
The syndrome decoding algorithm can be, for instance, in the case of Goppa's
codes, Patterson's algorithm (see part \ref{complexity}).

The McEliece or the Niederreiter schemes are not naturally invertible, i.e. if one starts from a random
element $y$ of $\mathbb{F}_2^n$ and a code $\mathcal{C}[n,k,d]$ that we are able to decode
up to $d/2$, it is almost sure that we won't be able
to decode $y$ into a codeword of $\mathcal{C}$. This comes from
the fact that the density of the whole space that is decodable
is very small.

\subsection{CFS signature scheme}

The idea of the CFS scheme is to find
parameters $[n,k,d]$ that make successful the strategy of picking up random elements
until one is able to decode it  with high probability.
  More precisely, given $M$ a message to sign and $h$ a hash-function with range
$\{0,1\}^{n-k},$ we try to find a way to build $s \in \mathbb{F}_2^n$ of given
weight $t$ such
that $h(M)=\widetilde{H}s^{T}.$ For $\mathcal{D} ecode(\cdot)$ a decoding
algorithm, the CFS scheme works as follows~:\\

\noindent [{\bf Key Generation}]$\ $
\begin{enumerate}
\item Select $n$, $k$ and $t$ according to the security parameter $\kappa$.

\item  Pick a random parity check matrix $\widetilde{H}$ of a $(n, k)$-binary Goppa code decoding $t$ errors.

  \item Choose a random $(n-k)\times (n-k)$ non-singular matrix $Q$, a random $n\times n$
  permutation matrix $P$ and a hash-function $h~: \{0, 1\}^\ast
  \longrightarrow \mathbb{F}_2^{n-k}$.

  \item The public key is $H= Q\widetilde{H}P$ and the private key is $(Q, \widetilde{H}, P)$.

  \item Set $t=\frac{n-k}{\log_2 n}, i=0$.
  \end{enumerate}
\noindent [{\bf Sign}]$\ $

  \begin{enumerate}
  \item $i \leftarrow i+1$
  \item $x^\prime = \mathcal{D} ecode_{\widetilde{H}}\left(Q^{-1}h(m\|i)\right)$
  \item if no $x^\prime$ was found go to 1
  \item output $(i, x^\prime P)$
  \end{enumerate}
\noindent [{\bf Verify}] Compute $s^\prime = H{x^\prime}^T$ and $s = h(m \| i)$.
  The signature is valid if $s $ and $s^\prime$ are equal.

\medskip

\indent We get at the end an $\{s,j\}$ couple, such that~: $$h(M\oplus j)=\widetilde{H}s^{T}.$$ Let us notice that we can suppose that $s$ has
 weight $t=[d/2].$
In \cite{Dallot}, a proof of security in the Random Oracle Model for a modified version of the CFS scheme is given. We use the modified CFS scheme described there, and named as \textsf{mCFS}, as a building block for our scheme. The \textsf{mCFS} scheme is explained next.

\subsection{Modified CFS signature scheme}

\smallskip

\noindent [{\bf Key Generation}]$\ $
\begin{enumerate}
\item Select $n$, $k$ and
  $t$ according to $\kappa$.
  \item Pick a random parity check matrix $\widetilde{H}$
  of a $(n, k)$-binary Goppa code decoding $t$ errors.
  \item Choose a random $(n-k)\times (n-k)$ non-singular matrix $Q$, a random $n\times n$
  permutation matrix $P$ and a hash-function $h~: \{0, 1\}^\ast
  \longrightarrow \mathbb{F}_2^{n-k}$.
  \item The public key is $H= Q\widetilde{H}P$ and the
  private key is $(Q, \widetilde{H}, P)$.
  \item Set $t=\frac{n-k}{\log_2 n}$.
  \end{enumerate}
\noindent [{\bf Sign}]$\ $

  \begin{enumerate}
  \item $i \stackrel{R}{\leftarrow} \{1,\dots,2^{n-k}\}$
  \item $x^\prime = \mathcal{D} ecode_{\widetilde{H}}\left(Q^{-1}h(m\|i)\right)$
  \item if no $x^\prime$ was found go to 1
  \item output $(i, x^\prime P)$
  \end{enumerate}

\noindent [{\bf Verify}]Compute $s^\prime = H{x^\prime}^T$ and $s = h(m \| i)$.
  The signature is valid if $s $ and $s^\prime$ are equals.

\section{Stern's protocol}\label{Stern}
\PARstart Stern's scheme is an interactive zero-knowledge protocol which aims
at enabling a {\em prover} $\P$ to identify himself to a {\em verifier} $\V$.

\medskip

Let $n$ and $k$ be two integers such that $n\geq k$. Stern's scheme
assumes the existence of a public $(n-k)\times n$ matrix
$\widetilde{H}$ defined over the two elements field
$\mathbb{F}_{2}$. It also assumes that an integer $t\leq n$ has been
chosen. For security reasons (discussed in \cite{Ste93}) it is
recommended that $t$ is chosen slightly below the
so-called Gilbert-Varshamov bound (see \cite{MS77}). The matrix
$\widetilde{H}$ and the weight $t$ are protocol parameters and may
be used by several (even numerous) different provers

\medskip

Each prover $P$ receives a $n$-bit secret key $\SK$ (also denoted by
$s$ if there is no ambiguity about the prover) of Hamming weight $t$
and computes a {\em public identifier} $\PK$ such that
$i_P=\widetilde{H}\SK^{T}$. This identifier is calculated once in
the lifetime of $\widetilde{H}$ and can thus be used for several
identifications. When a user $P$ needs to prove to $V$ that he is
indeed the person associated to the public identifier $\PK$, then
the two protagonists perform the following protocol where $h$
denotes a standard hash-function~:

\medskip

\noindent   [\textbf{Commitment Step}] $P$ randomly chooses $y\in \F_2^n$ and a permutation $\sigma$ of $\{1,2,\ldots,n\}.$
      Then $P$ sends to $V$ the commitments $c_{1}$, $c_{2}$ and $c_{3}$ such that ~:
      $$c_{1}=h(\sigma\|\widetilde{H}y^{T});\ c_{2}=h(\sigma(y));\ c_{3}=h(\sigma(y\oplus \SK)),$$
     where $h(a\|b)$ denotes the hash of the concatenation of the sequences $a$ and $b$.\\

\noindent    [\textbf{Challenge Step}] $V$ sends $b \in \{0,1,2\}$ to $P$.\\

\noindent    [\textbf{Answer Step}] Three possibilities~:
      \begin{itemize}
      \item if $b=0~:$ $P$ reveals $y$ and $\sigma.$
      \item if $b=1~:$ $P$ reveals $(y\oplus \SK)$ and $\sigma.$
      \item if $b=2~:$ $P$ reveals $\sigma(y)$ and $\sigma(\SK).$ \\
      \end{itemize}

\noindent    [\textbf{Verification Step}] Three possibilities~:
      \begin{itemize}
      \item if $b=0~:$ $V$ verifies that $c_{1},c_{2}$ are correct.
      \item if $b=1~:$ $V$ verifies that $c_{1},c_{3}$ are correct.
    \item if $b=2~:$ $V$ verifies that $c_{2},c_{3}$ are correct,
    and that the weight of $\sigma (s)$ is $t$.\\
      \end{itemize}

\noindent    [\textbf{Soundness Amplification Step}] Iterate the above steps until the expected security level is reached.
\medskip

During the fourth Step, when $b$ equals $1$, it can be noticed that
$\widetilde{H}y^{T}$ derives
    directly from $\widetilde{H}(y\oplus \SK)^{T}$ since we have~:
$$\widetilde{H}y^{T}=\widetilde{H}(y\oplus \SK)^{T}\oplus \PK =\widetilde{H}(y\oplus \SK)^{T}\oplus\widetilde{H}\SK^{T} \enspace .$$

As proved in \cite{Ste93}, the protocol is zero-knowledge and for a
round iteration, the probability that a dishonest person succeeds in
cheating is $(2/3)$. Therefore, to get a confidence level of
$\beta$, the protocol must be iterated a number of times $k$ such
that $(2/3)^{k}\leq \beta$ holds. When the number of iterations
satisfies the last condition, then the security of the scheme relies
on the NP complete problem SD.

By virtue of the so-called Fiat-Shamir Paradigm \cite{FS86}, it is
possible to convert Stern's Protocol into a signature
scheme, but the resulting signature size is long (about $140$-kbit long
for $2^{80}$ security). Notice that this is large in comparison with
classical signature schemes, but it is more or less close to the size of many
files currently used in everyday life.

\section{New Identity-based identification scheme from Stern-Niederreiter protocols}
\label{scheme}

\PARstart
We describe now the first code-based identity-based identification method. The prover is identifying herself to the verifier. Let $id_{S},id_{P}$ be
the prover and of the identifier identities respectively.

\medskip

\noindent [{\bf Master key generation}] Let $\mathcal{C},\, H,\,\widetilde{H}=QHP$ the output of the key generation algorithm of the CFS signature scheme in Section~\ref{CFS}.
Let $h$ a hash function mapping to $\{0,1\}^{n-k}.$
$\widetilde{H}$ is made public, but the decomposition of $\widetilde{H}$ is a secret of the authority.\\

\noindent [{\bf Key extraction}]
On inputs the the decomposition of $\widetilde{H}$ and the user's identity $id_{P}$ the goal of the key extraction algorithm is to output
$s\in E_{q,n,t}$ such that $h(id_{P})=\widetilde{H}s^{T}.$ However $h(id_{P})$ \textit{might not be} in the target of $x\rightarrow \widetilde{H}x^{T}.$
That is to say that  $h(id_{P})$ is \textit{not necessarily} in the space
of decodable elements of $\mathbb{F}_2^n$. That problem can be solved thanks to the
following algorithm. Given $\mathcal{D} ecode(\cdot)$ a decoding algorithm for the hidden code~:\\

  \begin{enumerate}
  \item $i \stackrel{R}{\leftarrow} \{1,\dots,2^{n-k}\}$
  \item $x^\prime = \mathcal{D} ecode_{\widetilde{H}}\left(Q^{-1}h(id_{P}\|i)\right)$
  \item If no $x^\prime$ was found go to 1
  \item output $(i, x^\prime P)$\\
  \end{enumerate}

We get at the end a couple $\{s,j\},$ such that  $h(id_{P} \| j)=\widetilde{H}s^{T}.$
We can note that we have $s$ of weight $t$ or less.\\

\noindent [{\bf Interactive identification}]
We use a slight derivation of Stern's protocol.
We suppose that the prover obtained a couple $\{s,j\}$
verifying $h(id_{P} \| j)=\widetilde{H}s^{T}.$
$h(id_{P} \| j)$ is set to be the prover's public key. Identification is then performed by modifying Stern's protocol with respect to the public key $h(id_{P} \| j)$. Details follow.\\

    [\textbf{Commitment Step}] $P$ chooses randomly any word $y$ of $n$ bits and a permutation $\sigma$ of $\{1,2,\ldots,n\}.$
      Then $P$ sends to $S~:c_{1},c_{2},c_{3},j$ such that ~:
      $$c_{1}=h(\sigma\|\widetilde{H}y^{T});\ c_{2}=h(\sigma(y));\ c_{3}=h(\sigma(y\oplus s))$$

    [\textbf{Challenge Step}] $S$ sends $b \in \{0,1,2\}$ to $P$.\\

    [\textbf{Answer Step}] Three possibilities~:
      \begin{itemize}
      \item if $b=0~:$ $P$ reveals $y$ and $\sigma.$
      \item if $b=1~:$ $P$ reveals $(y\oplus s)$ and $\sigma.$
      \item if $b=2~:$ $P$ reveals $\sigma(y)$ and $\sigma(s)$\\
      \end{itemize}

    [\textbf{Verification Step}] Three possibilities~:
      \begin{itemize}
      \item if $b=0~:$ $S$ verifies that the $c_{1},c_{2}$ received at the second round are correct.
      \item if $b=1~:$ $S$ verifies that the $c_{1},c_{3}$ received at the second round are correct.
    For $c_{1}$ we can note that $\widetilde{H}y^{T}$ derives
    directly from $\widetilde{H}(y\oplus s)^{T}$ by ~:
    $$\widetilde{H}y^{T}=\widetilde{H}(y\oplus s)^{T}\oplus\widetilde{H}s^{T}$$
      \item if $b=2~:$ $S$ verifies that the $c_{2},c_{3}$ received at the second round have really been honestly calculated,
    and that the weight of $s.\sigma$ is $t$.\\
      \end{itemize}

    [\textbf{Soundness Amplification Step}] Iterate the commitment, challenge, answer and verification steps until the expected security is reached.

\medskip

Thanks to the Fiat-Shamir heuristic \cite{FS86} it is possible
to derive an identity-based signature scheme from the above identity-based identification scheme. Since this is a well-known cryptographic result, we refer the reader to \cite{FS86,BNN} for details.

\section{Proving Security of  $\textsf{\textsf{mCFS}-Stern}$ IBI scheme}
\label{sec:proof}

\begin{theorem} The IBI scheme from Section \ref{scheme} is secure in the sense of
{\sf imp-pa}  if the BD and CD problems are hard to solve.
\end{theorem}

\begin{proof}
A security reduction is obtained by adapting the proofs by Dallot \cite{Dallot} and Stern \cite{Stern} to our setting.
We build the proof following a sequence of games \textbf{Game 0, Game 1,} $\ldots$
\textbf{Game 0} is the original attack game, i.e the standard {\sf imp-pa} game.
Successive games are obtained by small modifications of the preceding games, in such a way that the
difference of the adversarial advantage in consecutive games is easily quantifiable. To
compute this difference, the following lemma is used~:

\begin{lemma} \label{lemma:diff}
Let $X_i$,$X_{i+1}$, $B$ be events defined in some probability distribution, and suppose that
$X_i \wedge \neg B \Leftrightarrow X_{i+1} \wedge \neg B$. Then $| \prob{X_i}-\prob{X_{i+1}}|
\leq \prob{B}$.
\end{lemma}

  Let $q_h,q_E,q_C$ denote the maximum number of queries that adversary $\A$ makes to the hash,
  user keys and conversation oracles.

\medskip

  We want to
   show there exists adversaries $\C,\D$ that break the BD and CD problems
   respectively.

\medskip

  To answer hash, user key and conversation queries, three lists $\hlist,\UKlist$ and $\Llist$ are
  maintained. If there is no value associated with an entry in a list, we denote its output by $\perp$.
  The list  $\hlist$ consists of tuples of the form $(x,s)$ indexed by $(id,i)$, where $i$ is
   an index in $\{1,\ldots,2^{n-k}\}$, $id$ is an identity, and $\widetilde{H}s^T=x=h(id,i)$ if $x\neq \perp \neq s$. The list
   $\UKlist$, consists of entries of the form $(id,sk[id])$. The list $\Llist$ contains
    indexes $\Lambda(m)$ associated to a message $m$, for which the simulator is able to produce
     a signature on $h(m,\Lambda(m))$.

\smallskip

\noindent \textbf{Game 0.} This the standard \textsf{imp-pa} game. The master public and secret keys are obtained by running algorithm $\textsf{Gen}_{\textsf{mCFS}}(1^\kappa)$ In particular, the master public key $\widetilde{H}=QHP$ plus a hash-function $h~:\{0,1\}^*\rightarrow \F_2^{n-k}$, and the master secret key is $(Q,H,P)$, where $H \getsr \Goppa(n,k)$, $Q$ is a non-singular $(n-k)\times(n-k)$ matrix and $P$ is a $n\times n$ permutation matrix.  Therefore
$\prob{X_0}=\mathsf{Adv}^{\mathsf{imp-pa}}_{\IBIscheme,\mathcal{A}}$. \smallskip

\medskip

\noindent \textbf{Game 1.}(Simulation of hash and user key queries) We change the way in which hash and user key extraction
queries are answered. For hash queries of the form $(id,i)$,  there are two situations,
depending on whether $i = \Lambda (id)$. If this is the case, a decodable syndrome
$x=\widetilde{H}s^T$ is given as the output, and the corresponding code-word $s$ is
stored, i.e. $\hlist$ is updated with $(x,s)$ in the entry indexed by $(id,i)$. If $i \neq \Lambda (id)$
hash queries are simulated by taking a random element in $\F_2^{n-k}$, and then these queries are distributed as with a random oracle. Details are shown in Figure 1.

On the other hand, user key queries on $id$ are answered by choosing the special index
$\Lambda(id)$ at random, calling the hash oracle on $(id,\Lambda(id))$ and outputting
$(s,i)$ as the resulting user secret key. Details are shown in Figure 2.

 At the end of the simulation, the random oracle $h$  has output $q_h
+ q_E +1$ syndromes. Some of them are produced with the special index $i=\Lambda(id)$;
these syndromes are not distributed uniformly at random in $\F_2^{n-k}$, instead they
have been modified as to enable responding user secret key queries. It might be then the
case that adversary $\A$ queried $h$ on some pair $(id,j)$ such that later $j$ is set to
$\Lambda(id)$. This will cause an incoherence, since then the output $h(id,j)$ will be a
random syndrome, instead of a decodable syndrome. The latter happens with probability at
most $\frac{q_E}{2^{n-K}}$ (the indexes $\Lambda(id)$ are only defined when answering key
extraction queries). Therefore,
$$
|\Pr[X_0]-\Pr[X_1]| \leq \frac{q_E}{2^{n-k}}
$$

\noindent \textbf{Game 2.}(Changing the master key generation algorithm) The key generation algorithm is changed so that $H \leftarrow
\Binary(n,k)$. Then, $$ |\Pr[X_2]-\Pr[X_1]| \leq \mathsf{Adv}^{\mathsf{CD}}_{\mathcal{D}}
(n)
$$

\noindent where $\D$ is an algorithm that simulates the environment of Game 2 for $\A$ if
$H \leftarrow \Goppa(n,k)$ and outputs $d=1$ if $\A$ successfully impersonates the target
identity $id^\star$, and $d=0$ otherwise; and $\D$ simulates the environment of Game 3
for $\A$ if $H \leftarrow \Binary(n,k)$ and outputs $d=1$ if $\A$ successfully
impersonates the target identity $id^\star$, and $d=0$ otherwise. It is easy to see that
$$
\Pr[H\stackrel{R}{\leftarrow}\textsf{Goppa}(n, k)~:\mathcal{D}(H)=1] = \Pr[X_2].
$$
and $$ \Pr[H\stackrel{R}{\leftarrow}\textsf{Binary}(n, k)~:\mathcal{D}(H)=1] = \Pr[X_3].
$$

\begin{center}
    \begin{minipage}{0.8\linewidth}

        \begin{algorithm}[H]
          \SetLine
          \KwIn{A pair $(id, i)$}
          \KwOut{A syndrome $s$}
          $(s, x) \leftarrow \hlist(m, i)$\;
          \eIf{$i\not=\Lambda(id)$}{\label{debutfacile}
            \If{$s=\perp$}{
              $x\stackrel{R}{\leftarrow}\mathbb{F}_2^{n-k}$\;
              $ \hlist(id,i)\leftarrow (x,\perp)$\;
            }
            \Return{$h(id,i) = x$}\;
          }{\label{finfacile}
            \If{$x = \perp$}{
              $s\stackrel{R}{\leftarrow} \{w\in\mathbb{F}_2^n|\textsf{wt}(w)\le t\}$\;
              $x \leftarrow \widetilde{H}s^T$\;
              $\hlist(id,i) \leftarrow (x, s)$\;
            }
            \Return{$h(id,i) = x$}\;
          }
        \end{algorithm}
    \end{minipage}
    
    \label{alg:hash}

\end{center}
 \begin{center} Fig. 1 - Simulation of hash queries \end{center}

\medskip

\begin{center}
    \begin{minipage}{0.6\linewidth}

      \begin{algorithm}[H]
        \SetLine
        \KwIn{An identity $id$}
        \KwOut{A user secret key $(s,id)$}
        \If{$\Lambda(id) = \perp$}{
          $\Lambda(id)\stackrel{R}{\leftarrow}\{1,\ldots, 2^{n-k}\}$\;
        }
        $(x, s) \leftarrow h(id, \Lambda(id))$\;
        $i \leftarrow \Lambda(id)$\;
        $\Lambda(id) \leftarrow \perp$\;
        \Return{$sk[id]= (s,i)$}\;
      \end{algorithm}
    \end{minipage}

    \label{alg:sig}

\end{center}
 \begin{center} Fig. 2 - Simulation of user key queries \end{center}

\medskip

\noindent \textbf{Game 3.}(Guessing the target identity) A random index $j^+ \getsr \{1,\ldots,q_h+q_E+q_C\}$ is taken. The $j^+$-th hash query $(id^+,i^+)$ to is set to be $Q(x^{+})^T$, where $x^{+} \getsr \F_2^{n-k}$, i.e.
$h(id^+,i^+)=Q(x^{+})^T$. The probability
space is not modified, since $x^+ \getsr \F_2^{n-k}$ and $Q$ is  non-singular, and therefore $\Pr[X_3]=\Pr[X_2]$.

\smallskip

\medskip

\noindent \textbf{Game 4.}(Abort the game)

 Let
$(id^\star,i^\star)$ be the target identity and target index that $\A$ impersonates. If
$id^\star\neq id^+$ or $i^\star\neq i^+$ then the challenger aborts the game. Since Game 4 is obtained
by conditioning Game 3 on an independent event of probability $\frac{1}{q_H+q_E+1}$ we
obtain

$$
\Pr[X_4]=\frac{\Pr[X_3] }{q_H+q_E+1}
$$

\smallskip

\medskip

\noindent \textbf{Game 5.} (Answering conversation queries on the target identity
$id^\star$) We have to answer conversation queries on $id^\star$ without knowing the code
word $s^\star$ corresponding to $h(id^\star,i^\star)=x^\star$, i.e. $s^\star$ such that
$x^\star=\widetilde{H}s^+$ and $x^\star=Q(x^{+})^T$. We can answer these queries in expected polynomial time by using the algorithm
in Theorem 3 in~\cite{Stern}. Roughly, the algorithm uses a resettable simulation~\cite{GMR89}. At the
beginning of each iteration of the basic identification protocol, the algorithm chooses
at random one out of three cheating strategies, where each strategy allows to successfully
interact with a cheating verifier $\CV$ with probability $2/3$. In case the algorithm
can not successfully interact with $\CV$, it resets the adversary $\A$ for the current
round (see~\cite{Stern} for details). All in all, the probability space is not modified,
and then $\Pr[X_5]=\Pr[X_4]$.

\medskip

Theorem 1 in~\cite{Stern} implies that an adversary $\A$ impersonating the user with identity $id^\star$ when running $k$ rounds of the basic protocol and with advantage $(2/3)^k+\epsilon_1$ for a
non-negligible $\epsilon_1>0$, can be converted into a PPT algorithm computing $s^\star$ such that $\widetilde{H}(s^{\star})^T=x^\star$ with
probability $\epsilon_1^3/10$. A basic calculation shows that $(s^+)^T=P(s^\star)^T$ is a solution to the BD problem with inputs $H \getsr \Binary(n,k)$ and $x^+\getsr F_2^{n-k}$. Let $\C$ be an algorithm that simulates Game 5 for the
impersonating adversary  $\A$ using the input of the BD problem. Then,
$$\mathsf{Adv}^{\mathsf{BD}}_{\mathcal{C}} \geq \frac{\left(\Pr[X_5] -(2/3)^k \right)^3}{10}$$

\smallskip

Collecting all the probabilities

\begin{eqnarray*}
(2/3)^k+\epsilon &\leq & \mathsf{Adv}^{\mathsf{imp-pa}}_{\IBIscheme,\mathcal{A}} \\
&\leq & \frac{q_E}{2^{n-k}} + \mathsf{Adv}^{\mathsf{CD}}_{\mathcal{D}} (n)+ \Pr[X_5](q_h+q_E+1) \\
&\leq &  \frac{q_E}{2^{n-k}} +
\mathsf{Adv}^{\mathsf{CD}}_{\mathcal{D}} (n) + \\
&& \left((\mathsf{Adv}^{\mathsf{BD}}_{\mathcal{C}})^{1/3} + (2/3)^k\right)10^{1/3}(q_h+q_E+1)
\end{eqnarray*}

and then

\begin{eqnarray*}
\epsilon &\leq &  \frac{q_E}{2^{n-k}} + \mathsf{Adv}^{\mathsf{CD}}_{\mathcal{D}} (n) + \\
&& \left((\mathsf{Adv}^{\mathsf{BD}}_{\mathcal{C}})^{\frac{1}{3}} +  (1-\frac{1}{\sqrt[3]{10}})(\frac{2}{3})^k\right)10^{\frac{1}{3}}(q_h+q_E+1)
\end{eqnarray*}

\medskip

The latter equation can be read as follows~: a successful impersonating adversary with advantage $(\frac{2}{3})^k+\epsilon $ implies a successful adversary against the BD or CD problems. %\qed

\end{proof}

\section{Efficiency Analysis}
\label{sec:efficiency}
\PARstart We deal here with the security our protocol and its practicality.
Let us remind that in the case of Niederreiter's cryptosystem, its security
relies on the hardness of decoding of a linear
code (see section \ref{Pb}).

\subsection{Parameters and security of the scheme}

The protocol has two parts~: in the first part one inverts
the syndrome decoding problem for a matrix $\widetilde{H}$
in order to construct a private key for the prover
and in second part one applies Stern identification protocol
with the same matrix $\widetilde{H}$.
This shows that the overall parameters of the scheme
are equivalent to the security of the CFS scheme, since
the security of the Stern scheme with the same matrix
parameters is implicitly included in the signature scheme.

\medskip

In particular the scheme  has to fulfill two
imperative conditions~:
\begin{enumerate}
\item make the computation of $\{s,j\}$ (defined in advance)
 difficult without the knowledge of the
  description of $H,$
\item make the number of trials to determine
  the correct $j$ not too important in order to reduce the cost of
  the computation of $s$.
\end{enumerate}

\medskip

Following \cite{CFS} the Goppa $[2^m,2^m-tm,t]$ codes are a large
class of codes which are compatible with condition 2.
Indeed, for such a code, the proportion of the decodable syndromes is
about $1/t!$ (which is a relatively good proportion). We also have to
choose a relatively small $t.$

\medskip

The  $\{s,j\}$ production process will thus be iterated,
about $t!$ times before finding the correct $j.$ But each iteration
forces to compute $D(h(id_{P}\| j)).$

\medskip

The decoding of the Goppa codes consists of ~: \label{complexity}
\begin{itemize}
\item computing a syndrome ~: $t^{2}m^{2}/2$ binary operations;
\item computing a locator polynomial ~: $6t^{2}m$ binary operations;
\item computing its roots ~: $2t^{2}m^{2}$ binary operations.
\end{itemize}

\medskip

\indent We thus get a total cost for the computation of the prover's
private key of about~: $$t!t^{2}m^{2}(1/2+2+6/m)\ \textrm{binary operations}$$

\medskip

The cost of an attack by decoding thanks to the \textit{split syndrome
  decoding} is estimated to~: $2^{tm(1/2+o(1))}.$

\medskip

The choice of parameters will have to be pertinent enough to
conciliate cost and security. Although less important, some sizes have also to remain reasonable ~: the length
of $\{s,j\},$ the cost of the verification and the size of
$\widetilde{H}$ that is for a Goppa code~: $2^{m}tm.$

\medskip

Following \cite{CFS} we can for example take $t=9$ and $m=16.$
The cost of the signature stays then relatively reasonable for a
security of about $2^{80}.$ The others sizes remain in that context
very acceptable.

\subsection{Practical values}

The big difference when using the parameters associated
to the CFS scheme is that the code used is very long,
$2^{16}$ against $2^9$ for the basic Stern scheme,
it dramatically develops communication costs.

\medskip

In the next table we sum up for the parameters $m=16$, $t=9$
the general parameters of the IBI and IBS schemes.\\

\begin{center}
  \begin{tabular}{|c|c|c|}
    \hline
     public key & private key & matrix size \\
    \hline
    $tm$ & $tm$&$2^mtm$\\
    \hline
    \hline
    144&144& 1 Mo\\
 \hline
  \end{tabular}
\end{center}

\begin{center}
  \begin{tabular}{|c|c|}
    \hline
      communication cost & key generation\\
    \hline
    $\approx 2^m\times \#rounds$ & \\
    \hline
    \hline
    500 Ko (58 rounds) & 1 s\\
 \hline
  \end{tabular}
\end{center}

\begin{center} Practical values for the IBI scheme~: $m=16,t=9$
\end{center}

\noindent

\begin{center}
  \begin{tabular}{|c|c|}
    \hline
      signature length & key generation\\
    \hline
    $\approx 2^m\times \#rounds$ & \\
    \hline
    \hline
     2.2 Mo (280 rounds) & 1 s\\
 \hline
  \end{tabular}
\end{center}

\begin{center} Practical values for the IBS scheme~: $m=16,t=9$
\end{center}

\medskip
\noindent {\bf Reduction of the size of the public matrix~:} At the difference of
a pure signature scheme in which one wants to be able to
sign fast, in our scheme the signature is only computed once
for sending it to the prover, hence the time for
signing may be judged less
determinant and a longer time of signature may be accepted
at the cost of reducing (a little) the parameters of the
public matrix.

\section{Conclusion}
\label{sec:conclusion}

\PARstart In this paper we present and prove secure a new identity-based identification scheme based on error-correcting codes. Our scheme combines two well known schemes by Courtois-Finiasz-Sendrier and Stern. It inherits some of their practical weaknesses, such as large system parameters. Interestingly the new scheme is one of the very few existing alternatives to number theory for identity-based cryptography,  and we hope that it boosts future research on this area.

\bibliographystyle{abbrv}

\end{document}